\documentclass[draft,onecolumn]{IEEEtran}
\IEEEoverridecommandlockouts
\usepackage{cite}
\usepackage{amsmath,amssymb,amsfonts,amsthm}
 \usepackage{algorithm}
 \usepackage{algpseudocode}
\usepackage{graphicx}
\usepackage{textcomp}
\usepackage{xcolor}
\usepackage{setspace} 
\usepackage{tikz}
\usepackage{url}
\usetikzlibrary{positioning}
\usetikzlibrary{calc,positioning}
 \usetikzlibrary{arrows.meta,positioning}

\usepackage{blkarray}
\usepackage{comment}

\begin{document}

\newtheorem{theorem}{Theorem}
\newtheorem{definition}{Definition}
\newtheorem{corollary}{Corollary}
\newtheorem{proposition}{Proposition}
\newtheorem{lemma}{Lemma}
\newtheorem{example}{Example}
\newtheorem{remark}{Remark}
\newtheorem{construction}{Construction}
\newcommand{\F}{\mathbb{F}}
\newcommand{\wt}{\mathrm{wt}}
\newcommand{\Gopt}{G_{\mathrm{opt}}}

\newif\ifarxiv \arxivtrue

\ifarxiv
  \newcommand{\arxiv}[1]{#1}\newcommand{\conf}[1]{}
\else
  \newcommand{\arxiv}[1]{}\newcommand{\conf}[1]{#1}
\fi

\title{On Optimal Encodings and Systematicity in Function-Correcting Codes}

\author{\IEEEauthorblockN{Charul Rajput\IEEEauthorrefmark{1}, Kanchana Lokshmii Jagatti\IEEEauthorrefmark{2}, B. Sundar Rajan\IEEEauthorrefmark{2}}

\IEEEauthorblockA{\IEEEauthorrefmark{1}University of Bristol Mumbai Enterprise Campus, Powai, Mumbai, India \\
Email: charul.rajput@mumbai.bristol.ac.uk}
\linebreak
\IEEEauthorblockA{\IEEEauthorrefmark{2}Department of Electrical Communication Engineering,
Indian Institute of Science, Bengaluru, India \\
Emails:  kanchanal@iisc.ac.in, bsrajan@iisc.ac.in}}

%

\maketitle

\begin{abstract}
Function-correcting codes (FCCs) protect the value of a function $f$ of the message against $t$ errors. In the original formulation of Lenz et al. (IEEE Trans. Inf. Theory, 2023), the encoding is systematic, and the lower bound of $2t$ on the redundancy relies on this form. Recently, D. Ho (arXiv, 2026) showed, for linear functions and linear encodings, that systematicity can cost redundancy. We begin with the OR function to show that the cost is not confined to the linear setting: a non-systematic encoding attains redundancy $1$ while every systematic encoding needs $2t$. For a linear function $f$ and a fixed linear code $C$, let $d_f$ denote the minimum distance between codewords of messages with different function values. Different generator matrices of $C$ assign different codewords to the messages and can give different values of $d_f$. We study which generator matrix of $C$ gives the largest $d_f$. We give an algorithm that constructs an optimal generator matrix and determines the optimal value as the weight of a codeword in a greedy basis of $C$. We then characterize, in terms of information sets, when a generator matrix in systematic form attains this optimum, and give a necessary condition that is checked on the low-weight codewords of $C$ alone. \arxiv{ For two-valued functions, we drop both linearity and systematicity. Using the minimality of initial segments of the simplicial order with respect to Hamming neighbourhoods, we show that a non-systematic $(f,t)$-FCC of length $n$ exists if and only if a condition depending on $f$ only through the size of its smaller preimage holds. For $t=1$, redundancy $1$ is sufficient for every nonconstant two-valued function on $\mathbb{F}_2^k$ with $k\ge 10$, and we list the cases for $k\le 9$. For general $t$, redundancy $1$ is sufficient once $k\ge\lceil\frac83(2t+1)^2\rceil-1$, and for each redundancy $s<2t$ we define the least $k$ beyond which redundancy $s$ is sufficient for all such functions and give an upper bound on it.}
\end{abstract}

%




\section{Introduction}
\label{intro}

Function-correcting codes (FCCs) were introduced by Lenz \emph{et al.}~\cite{LBWY2023} for settings in which the message is transmitted and the value of a function $f$ of the message must be protected against errors. In their formulation, a message $u\in\F_2^k$ is encoded systematically as $\mathrm{Enc}(u)=(u,p(u))$ with $p(u)\in\F_2^r$, and the encoding is an $(f,t)$-FCC if
\begin{equation}\label{eq:fcc}
f(u)\neq f(v)\implies d\big(\mathrm{Enc}(u),\mathrm{Enc}(v)\big)\ge 2t+1
\end{equation}
for all messages $u,v$. Condition~\eqref{eq:fcc} protects only the function value against $t$ errors: two messages with the same function value may be encoded at distance one. The systematic form makes the encoding one-to-one, so that every message can be identified from its codeword when no error occurs, but it gives no protection to the message itself. The systematic formulation has been retained in subsequent work on bounds and constructions~\cite{PR2025,LS2025,GXZZ2025}. In~\cite{RRHH2025}, the message is also protected, through an additional distance requirement between all pairs of codewords. The notation used throughout the paper is collected at the end of this section.

Since no protection of the message is required in~\eqref{eq:fcc}, the systematic form serves only to make the encoding one-to-one, and any one-to-one encoding would serve this purpose. This raises the question of whether the systematic form costs redundancy. The following example shows that it can.

\begin{example}\label{ex:or}
Let $f:\F_2^k\to\F_2$ be the OR function, that is, $f(0)=0$ and $f(u)=1$ for $u\neq 0$, and let $t\ge 1$. For a systematic encoding, the messages $0$ and $e_1$ have different function values and their systematic parts are at distance one, so $d(p(0),p(e_1))\ge 2t$ and $r\ge 2t$. The parity map $p(0)=0^{2t}$, $p(u)=1^{2t}$ for $u\neq 0$, where $0^{2t}$ and $1^{2t}$ denote the all-zero and all-one vectors of length $2t$, attains $r=2t$.

Now let $\mathrm{Enc}:\F_2^k\to\F_2^n$ be any one-to-one encoding. Condition~\eqref{eq:fcc} only requires the $2^k-1$ codewords of nonzero messages to lie at distance at least $2t+1$ from $\mathrm{Enc}(0)$. Taking $\mathrm{Enc}(0)=0$ without loss of generality, such an encoding exists if and only if we have sufficient vectors of weight at least $2t+1$ in $\mathbb{F}_2^n$, i.e.,
\[
2^n-V(n,2t)\ge 2^k-1.
\]
For $k=4$ and $t=1$, this holds for $n=5$ but not for $n=4$. An optimal encoding maps $0000\mapsto 00000$ and the $15$ nonzero messages bijectively onto the $15$ vectors of weight $3$ or $4$ in $\F_2^5$. This gives length $5$, against length $6$ for the systematic encoding as shown in Tab \ref{tab:or}. \arxiv{More generally, $n=k+1$ suffices whenever $V(k+1,2t)\le 2^k+1$, which holds for every fixed $t$ and all sufficiently large $k$. The redundancy then drops from $2t$ to $1$, and the difference $2t-1$ grows without bound with $t$.} 

\begin{table}[t]
\centering
\caption{Encodings of the OR function on $\F_2^4$ for $t=1$}
\label{tab:or}
\begin{tabular}{cccc}
\hline
Message $u$ & $f(u)$ & Systematic, $n=6$ & Non-systematic, $n=5$ \\
\hline
$0000$ & $0$ & $000000$ & $00000$ \\
\hline
$1000$ & $1$ & $100011$ & $01111$ \\
$0100$ & $1$ & $010011$ & $10111$ \\
$0010$ & $1$ & $001011$ & $11011$ \\
$0001$ & $1$ & $000111$ & $11101$ \\
$1100$ & $1$ & $110011$ & $11001$ \\
$1010$ & $1$ & $101011$ & $10101$ \\
$1001$ & $1$ & $100111$ & $10011$ \\
$0110$ & $1$ & $011011$ & $01101$ \\
$0101$ & $1$ & $010111$ & $01011$ \\
$0011$ & $1$ & $001111$ & $00111$ \\
$1110$ & $1$ & $111011$ & $11100$ \\
$1101$ & $1$ & $110111$ & $11010$ \\
$1011$ & $1$ & $101111$ & $10110$ \\
$0111$ & $1$ & $011111$ & $01110$ \\
$1111$ & $1$ & $111111$ & $11110$ \\
\hline
\end{tabular}
\end{table}
\end{example}

The lower bound $r\ge 2t$ of~\cite[Cor.~1]{LBWY2023}, valid for every $f$ with $|\mathrm{Im}(f)|\ge 2$, relies on the systematic form. Without the systematic form, this argument gives only $r\ge 1$ for $t\ge 1$. If $n=k$, the encoding is a bijection onto $\F_2^k$, and since $f$ is nonconstant, some two codewords at distance one carry different function values. Example~\ref{ex:or} shows that this weaker bound can be attained, so for the OR function the bound $2t$ exceeds the optimal redundancy by $2t-1$. For a linear function and a linear encoding, the bound $r\ge 2t$ still holds by~\cite[Cor.~4.9]{Ho2026}, so a redundancy below $2t$ requires a nonlinear encoding. A nonlinear encoding can go below it even when $f$ is linear as shown in the following Example \ref{ex:lin2}.

\begin{example}\label{ex:lin2}
Let $f:\F_2^8\to\F_2$, $f(u)=u_1$, and $t=1$. A systematic encoding needs $r\ge 2$. Each of the two classes $f^{-1}(0)$ and $f^{-1}(1)$ has $128$ messages. In $\F_2^9$ there are $130$ vectors of weight at most $3$ and $130$ vectors of weight at least $6$. Map $f^{-1}(0)$ injectively into the first set and $f^{-1}(1)$ injectively into the second. Any two vectors from different sets are at distance at least $3$, so this encoding of length $9$ protects the function value against one error with $r=1$. \arxiv{Section~\ref{sec:twoval} determines, for every nonconstant two-valued $f$ and every $t$, exactly when such an encoding exists with a given length $n$.}
\end{example}

\subsection{Related Work}

Since~\cite{LBWY2023}, work on FCCs has followed two directions. The first retains the Hamming metric and sharpens the redundancy: bounds and constructions for general and for linear functions are given in~\cite{PR2025}, the optimal redundancy is determined for Hamming weight distribution functions in~\cite{GXZZ2025}, the problem is studied over larger fields in~\cite{LS2025}, and functions with locally bounded variation are treated in~\cite{RRHH2025ITW25}. The second changes the channel: symbol-pair and $b$-symbol read channels~\cite{XLC2024,SSY2025,SR2025,VSS2025}, the Lee metric~\cite{VS2025,HUR2025}, and the homogeneous distance~\cite{LL2025}. In~\cite{RRHH2025}, the message itself is protected in addition to the function value, and~\cite{RRHH2026} generalizes the function to an arbitrary partition of the message space. All of these works keep the systematic form of the encoding. The question of which encoding of a given code best protects a part of the message is older than FCCs: it is the unequal error protection problem of Masnick and Wolf~\cite{MW1967}, for which Dunning and Robbins~\cite{DR1978} determined the optimal encodings of a linear code and van Gils~\cite{vG1983} gave length bounds. \arxiv{ For non-systematic encodings of two-valued functions, the relevant extremal result is the vertex-isoperimetric theorem of Harper~\cite{H1966} and its extension to arbitrary radius by R\"aty~\cite{R2020}.}

The cost of systematicity was first studied by Ho~\cite{Ho2026}, for linear functions and linear encodings. The central parameter there is the function-separation distance $\delta_{\mathrm{Enc}}(f)$, the minimum distance between codewords of messages with different function values; the classical unequal error protection (UEP) parameters~\cite{MW1967,DR1978,vG1983} are function-separation distances of coordinate and projection maps. 

Over all linear encodings, the least redundancy achieving separation distance $\delta$ depends on $f$ only through $\operatorname{rank} f$~\cite[Th.~4.4]{Ho2026}. Requiring systematicity relative to a fixed basis can increase it. For $\delta\le 3$ the systematic optimum is determined exactly~\cite[Prop.~5.1, Th.~5.6]{Ho2026}: the gap is zero for $\delta\le 2$, while at $\delta=3$ it vanishes exactly when a projective condition on $f$ holds and is otherwise unbounded~\cite[Cor.~5.9, Cor.~5.12]{Ho2026}. Example~\ref{ex:or} lies outside this setting, since the OR function is not linear and the optimal encoding there is nonlinear.

\subsection{Contributions and Organization}

For a linear function $f$ and a code $C$ in systematic form, linear FCCs were characterized in~\cite[Sec.~VII]{RRHH2025} through the coset code $C/D_f$, where $D_f=\{(u,p)\in C : u\in\ker f\}$: the function values are separated by a distance $\delta$ if and only if the minimum distance between distinct cosets of $D_f$ in $C$, denoted $d(C/D_f)$, is at least $\delta$~\cite[Def.~18]{RRHH2025}. The systematic form is chosen there since every linear code is equivalent to one in systematic form, but this choice fixes the encoding of each code. Different generator matrices of the same code assign different codewords to the messages. For an arbitrary linear encoding, the subcode $\mathrm{Enc}(\ker f)$ plays the role of $D_f$, and $\delta_{\mathrm{Enc}}(f)=\min\{\mathrm{wt}(c):c\in C\setminus \mathrm{Enc}(\ker f)\}$~\cite[Th.~3.6]{Ho2026}. Ho also identifies this quantity with the first relative generalized Hamming weight of $C$ with respect to $\mathrm{Enc}(\ker f)$. For the encoding $u\mapsto uG$ and a linear function $f$ with $U=\ker f$, we write $d_f(G)$ for the minimum distance between codewords of messages with different function values; thus $d_f(G)$ is the minimum weight of $C\setminus UG$, and the choice of $G$ determines which subcode of $C$ plays the role of $UG$. Example~\ref{ex:nosys} shows that no generator matrix in systematic form need attain the largest value of $d_f(G)$.

\begin{example}\label{ex:nosys}
Let $C=\{000,100,011,111\}$ be the $[3,2]$ binary linear code and let $f:\F_2^2\to\F_2$ be given by $f(u_1,u_2)=u_1+u_2$, so that $\ker f=\{00,11\}$. Up to a permutation of the coordinates, the generator matrices of $C$ in systematic form are
\[
G_1=\begin{bmatrix}1&0&0\\0&1&1\end{bmatrix}
\qquad\text{and}\qquad
G_2=\begin{bmatrix}1&0&1\\0&1&0\end{bmatrix}.
\]
In both cases, the message $11$ is encoded to $111$. The messages $00$ and $11$, which have the same function value, are therefore encoded to $000$ and $111$, while each of $10$ and $01$ is encoded to a codeword at distance $1$ from $000$ or $111$. Hence $d_f(G_1)=d_f(G_2)=1$.
 
On the other hand, for
\[
G=\begin{bmatrix}1&1&1\\0&1&1\end{bmatrix},
\]
the messages $00,11$ are encoded to $000,100$, and the messages $10,01$ are encoded to $111,011$. Every codeword in $\{000,100\}$ is at distance at least $2$ from every codeword in $\{111,011\}$, so $d_f(G)=2$. This is the largest possible value: the messages $00$ and $11$ are encoded to $000$ and a nonzero codeword $x$, and the other two messages are encoded to the remaining two nonzero codewords, one of which has weight at most $2$ and is therefore at distance at most $2$ from $000$. Thus, no generator matrix in systematic form achieves the optimal $d_f$ for this code and function.
\end{example}

In~\cite{Ho2026}, the code is free: the optimization ranges over all codes of a given length. In this paper, we take a complementary view, in which the linear code $C$ is fixed, for instance because of its parameters or its decoding algorithm, and the encoding of $C$ is to be chosen. All generator matrices of $C$ give codes with the same $n$, $k$, and $d$, but they can differ as FCCs. For the separation vector, the vector whose $i$th entry is the separation distance of the $i$th message coordinate, the analogous question was considered by Dunning and Robbins~\cite{DR1978}. They showed that every linear code has a generator matrix that is optimal for all message positions simultaneously, and that one is obtained by repeatedly choosing a minimum-weight codeword outside the span of those already chosen. \arxiv{Examples~\ref{ex:or} and~\ref{ex:lin2} lead to a second question, which lies outside the linear setting of~\cite{Ho2026}: given a two-valued function $f$ and an arbitrary one-to-one encoding, what is the least length $n$ for which an $(f,t)$-FCC exists? In this setting neither the code nor the encoding is linear, and the answer is determined by the minimum size of the radius-$2t$ neighbourhood of a set of $\min{(|f^{-1}(0)|,|f^{-1}(1)|)}$ vectors in $\F_2^n$. This minimum is given by the vertex-isoperimetric theorem of Harper~\cite{H1966} and its extension to arbitrary radius by R\"aty~\cite{R2020}.}

Our contributions are as follows.
\begin{itemize}
\item For a given code $C$ and linear function $f$, we give an algorithm that constructs a generator matrix $\Gopt$ with $d_f(\Gopt)=d_f^{*}=\max_G d_f(G)$, and show that $d_f^{*}=\mathrm{wt}(b_{r+1})$, where $b_1,\dots,b_k$ is the greedy basis of~\cite{DR1978} and $r=\dim\ker f$. 
\item We characterize when a generator matrix in systematic form attains $d_f^{*}$, in terms of ordered information sets of $C$. The condition needs to be checked only on the codewords of weight less than $d_f^{*}$, and it gives a simple necessary condition.
\arxiv{\item For a nonconstant $f:\F_2^k\to\F_2$, let $M=\min\{|f^{-1}(0)|,|f^{-1}(1)|\}$ be the size of the smaller preimage. We show that a non-systematic $(f,t)$-FCC of length $n$ exists if and only if the set of vectors within distance $2t$ of the first $M$ vectors of $\F_2^n$ in the simplicial order has at most $2^n-2^k+M$ elements (Theorem~\ref{thm:twoval}). 
\item For $t=1$, we list the class sizes $M$ for which redundancy $1$ suffices when $k\le 9$ (Table~\ref{tab:r1}), and show that for $k\ge 10$ every nonconstant two-valued function admits a non-systematic $(f,1)$-FCC with redundancy $1$, which is optimal (Proposition~\ref{prop:k10}). For a nonzero linear $f$, redundancy $1$ is possible if and only if $k\ge 8$.
\item For general $t$, redundancy $1$ is achievable and optimal for every nonconstant two-valued $f$ once $k\ge\lceil\frac83(2t+1)^2\rceil-1$ (Theorem~\ref{thm:larget}).
\item We define the redundancy-$s$ threshold $\kappa_s(t)$ as the least $\kappa$ such that, for every $k\ge\kappa$, every nonconstant two-valued function on $\F_2^k$ admits a non-systematic $(f,t)$-FCC of length $k+s$. Then we provide an upper bound on this threshold (Theorem~\ref{thm:kappa}).}
\end{itemize}

The rest of the paper is organized as follows. Section~\ref{sec:main} considers linear encodings of a fixed linear code: Section~\ref{sec:opt} introduces the kernel subcode of a generator matrix, Section~\ref{subsec:opt} presents the algorithm and establishes its optimality, and Section~\ref{subsec:sys} characterizes when the optimum is attained by a generator matrix in systematic form. \arxiv{Section~\ref{sec:twoval} considers non-systematic encodings of two-valued functions, first for $t=1$ and then for general $t$. } Section~\ref{conclusion} concludes the paper.

\emph{Notation.} For a prime power $q$, $\F_q$ denotes the finite field with $q$ elements, and an $[n,k]$ code over $\F_q$ is a $k$-dimensional subspace of $\F_q^n$. For $x,y\in\F_q^n$, $d(x,y)$ is the Hamming distance and $\wt(x)=d(x,0)$ the Hamming weight\arxiv{, and $\mathrm{supp}(x)$ is the set of nonzero coordinates of $x$; for $S,T\subseteq\F_q^n$, $d(S,T)=\min\{d(x,y):x\in S,\ y\in T\}$}. The minimum distance of a code $C$ is denoted by $d$, and $\langle S\rangle$ is the subspace spanned by a set $S$ of vectors. The vector $e_i$ is the $i$th standard basis vector, and $V(n,\rho)=\sum_{i=0}^{\rho}\binom{n}{i}$ is the number of vectors of weight at most $\rho$ in $\F_2^n$.

\section{Main results}
\label{sec:main}

For a given code $C$, there are many generator matrices $G' = AGP$, where $A$ is an invertible $k\times k$ matrix and $P$ is a permutation matrix, and they generate $C$ or codes equivalent to it. All of these codes have the same parameters $n$, $k$ and $d$, and they perform identically as error-correcting codes. As function-correcting codes, however, they can differ. The matrix $A$ changes which message is mapped to which codeword, and hence which codewords lie in the same block of the partition on the message space $\mathbb{F}_q^k$ induced by $f$. As a result, for a given function $f$, some of these generator matrices achieve a larger $d_f$ than others. It is natural to look for a generator matrix in systematic form $[I_k \mid P]$, since the message then appears directly in the codeword. However, a generator matrix in systematic form that achieves the largest possible $d_f$ need not exist, as demonstrated in Example~\ref{ex:nosys}.

Motivated by Example~\ref{ex:nosys}, in the rest of this section we do not restrict to generator matrices in systematic form. Instead, for a given code and a given linear function, we find a generator matrix of $C$ that achieves the optimal $d_f$. We then identify conditions under which a generator matrix in systematic form with optimal $d_f$ exists.

\subsection{Kernel Subcode of a Generator Matrix}\label{sec:opt}
 
Let $C$ be an $[n,k]$ linear code over $\F_q$ and let $f:\F_q^k\to\F_q^l$ be a linear function. We have a subspace
\[
U=\ker f,\qquad \dim U = r = k-\operatorname{rank} f,
\]
so that $r=k-l$ when $f$ is surjective. The messages are partitioned into the cosets of $U$, and two messages $u,u'$ satisfy $f(u)\neq f(u')$ if and only if $u-u'\notin U$. For a generator matrix $G$ of $C$ with encoding $u\mapsto uG$, we call
\[
D_G = UG = \{uG : u\in U\}
\]
the \textit{kernel subcode} of $G$. It is an $r$-dimensional subcode of $C$, and the cosets of $U$ are mapped to the cosets of $D_G$ in $C$. We denote by $d_f(G)$ the minimum distance between codewords $uG$ and $u'G$ with $f(u)\neq f(u')$, and by $d_f^{*}=\max_G d_f(G)$ its largest value over all generator matrices $G$ of $C$.
 
In \cite[Sec.~VII]{RRHH2025}, for a linear function $f$ and a linear code in systematic form, the subcode $D_f=\{c=(u,p)\in C \mid u\in \ker f\}$ is considered, and the function-correcting condition is expressed through the distance between the cosets of $D_f$ in $C$. Here $G$ is an arbitrary generator matrix of $C$, not necessarily in systematic form, and the kernel subcode $D_G$ plays the same role: under the encoding $u\mapsto uG$, messages with different function values are mapped to codewords in different cosets of $D_G$, and conversely. Hence $d_f(G)$ is the minimum distance between codewords lying in different cosets of $D_G$. The following identity also appears in~\cite[Th.~3.6]{Ho2026}. We give a short proof in our notation.

\begin{lemma}\label{lem:df}
$d_f(G)=\min\{\wt(x) : x\in C\setminus D_G\}$.
\end{lemma}
\begin{proof}
Since $G$ has full rank, $f(u)\neq f(u')$, i.e., $u-u'\notin U$, if and only if $uG-u'G\notin D_G$. Hence
\[
d_f(G)=\min\{\wt(uG-u'G) : u,u'\in\F_q^k,\ uG-u'G\notin D_G\}.
\]
Every difference $uG-u'G$ in this set is a codeword in $C\setminus D_G$. Conversely, every $x\in C\setminus D_G$ is of the form $x=uG-0G$ with $u'=0$. Therefore $d_f(G)=\min\{\wt(x) : x\in C\setminus D_G\}$.
\end{proof}
 
 \subsection{Generator Matrix with Optimal $d_f$ for a Given Code and Function}\label{subsec:opt}

Finding the minimum distance, and hence the low-weight codewords, of a linear code is NP-hard in general. Assuming that the low-weight codewords of $C$ are available, we construct a basis of $C$ with the following greedy algorithm, as in~\cite{DR1978}.
\begin{itemize}
\item Start by selecting $b_1$ as a codeword of weight $d$, the minimum distance of $C$.
\item Take the next low-weight codeword which is independent of $b_1$, i.e.,
\[
b_2 = \text{a minimum-weight codeword in } C\setminus\langle b_1\rangle .
\]
\item In general, $b_{i+1}$ is a minimum-weight codeword in $C\setminus\langle b_1,\dots,b_i\rangle$. Do this $k$ times to find a basis $\{b_1,b_2,\dots,b_k\}$ of $C$.
\end{itemize}
Let $B$ be the $k\times n$ matrix
\[
B=\begin{bmatrix} b_1\\ b_2\\ \vdots\\ b_k\end{bmatrix}.
\]
Since $b_{i+1}\notin\langle b_1,\dots,b_{i-1}\rangle$, it was also a candidate at step $i$, and therefore
\begin{equation}\label{eq:mono}
\wt(b_1)\le \wt(b_2)\le\cdots\le\wt(b_k).
\end{equation}
 
Now we find a generator matrix $\Gopt$ which maps the elements of a basis $\{u_1,u_2,\dots,u_r\}$ of $U$ to $\{b_1,b_2,\dots,b_r\}$, so that the $r$ basis codewords of smallest weight are images of $\ker f$. To do that, extend $\{u_1,\dots,u_r\}$ to a basis $\{u_1,\dots,u_r,v_{r+1},\dots,v_k\}$ of $\F_q^k$ and form the $k\times k$ invertible matrix
\[
B_U=\begin{bmatrix} u_1\\ \vdots\\ u_r\\ v_{r+1}\\ \vdots\\ v_k\end{bmatrix}.
\]
Then $u_i=e_iB_U$ and $b_i=e_iB$, where $e_i$ is the $i$th standard basis vector of $\F_q^k$. If
\[
\Gopt = B_U^{-1}B,
\]
then for any $1\le i\le r$,
\[
u_i\Gopt = u_iB_U^{-1}B = e_iB_UB_U^{-1}B = e_iB = b_i .
\]
Since $B_U$ is invertible and the rows of $B$ form a basis of $C$, $\Gopt$ is a generator matrix of $C$, and its kernel subcode is $D_{\Gopt}=\langle b_1,\dots,b_r\rangle$. The complete procedure is given in Algorithm~\ref{alg:opt}.

\begin{algorithm}[t]
\caption{Generator matrix with optimal $d_f$ for a given code $C$ and function $f$}\label{alg:opt}
\begin{algorithmic}[1]
\Require An $[n,k]$ linear code $C$ over $\F_q$; a basis $\{u_1,\dots,u_r\}$ of $U=\ker f$
\Ensure A generator matrix $\Gopt$ of $C$ with $d_f(\Gopt)=\wt(b_{r+1})$
\For{$i=1$ to $k$}
    \State $b_i \gets$ a minimum-weight codeword in $C\setminus\langle b_1,\dots,b_{i-1}\rangle$
\EndFor
\State $B \gets$ the $k\times n$ matrix with rows $b_1,\dots,b_k$
\State Extend $\{u_1,\dots,u_r\}$ to a basis $\{u_1,\dots,u_r,v_{r+1},\dots,v_k\}$ of $\F_q^k$
\State $B_U \gets$ the $k\times k$ matrix with rows $u_1,\dots,u_r,v_{r+1},\dots,v_k$
\State $\Gopt \gets B_U^{-1}B$
\State \Return $\Gopt$ and $d_f(\Gopt)=\wt(b_{r+1})$
\end{algorithmic}
\end{algorithm}
 
\begin{theorem}\label{thm:opt}
For the generator matrix $\Gopt$ obtained from Algorithm~\ref{alg:opt},
\[
d_f(\Gopt)=\wt(b_{r+1}),
\]
and $d_f(G)\le \wt(b_{r+1})$ for every generator matrix $G$ of $C$. Hence $d_f^{*}=\wt(b_{r+1})$, and $\Gopt$ gives an optimal encoding.
\end{theorem}
\begin{proof}
We have $D_{\Gopt}=\langle b_1,\dots,b_r\rangle$. By the greedy choice, $b_{r+1}$ is a minimum-weight codeword in $C\setminus\langle b_1,\dots,b_r\rangle$, so Lemma~\ref{lem:df} gives $d_f(\Gopt)=\wt(b_{r+1})$.
 
Let $G$ be any generator matrix of $C$, with kernel subcode $D_G$ of dimension $r$. The $r+1$ independent codewords $b_1,\dots,b_{r+1}$ cannot all lie in $D_G$, so some $b_j\notin D_G$ with $j\le r+1$. By Lemma~\ref{lem:df} and \eqref{eq:mono}, $d_f(G)\le\wt(b_j)\le\wt(b_{r+1})$.
\end{proof}

\begin{remark}
The choice of the extension vectors $v_{r+1},\dots,v_k$ does not affect the messages in $U$, since $u_i\Gopt=b_i$ for $1\le i\le r$. It only changes the codewords assigned to the messages outside $U$. 
\end{remark}
 
\begin{example} \label{ex:73}
Consider a linear function $f:\F_2^3\to\F_2^2$ with $\ker f=\{000,110\}$. Let $U=\langle 110\rangle=\ker f$, so $\dim U = r = 1$. Let $C$ be the $[7,3]$ binary linear code with generator matrix
\[
G=\begin{bmatrix}
0&0&1&1&1&0&1\\
0&1&0&0&0&0&1\\
1&1&1&0&0&1&0
\end{bmatrix}.
\]
The codewords of this code have weights $0,2,4,4,4,4,4,6$, so its minimum distance is $d=2$.
 
\emph{Greedy basis.} We have $b_1=0100001$ with $\wt(b_1)=2$, and $b_2=0011101$ with $\wt(b_2)=4$. Then
\[
\langle b_1,b_2\rangle=\{0000000,\,0100001,\,0011101,\,0111100\},
\]
and $b_3$ is a codeword of lowest weight in $C\setminus\langle b_1,b_2\rangle$, namely $b_3=1110010$. So
\[
B=\begin{bmatrix}
0&1&0&0&0&0&1\\
0&0&1&1&1&0&1\\
1&1&1&0&0&1&0
\end{bmatrix}.
\]
Since $r=1$, Theorem~\ref{thm:opt} gives $d_f(\Gopt)=\wt(b_{r+1})=\wt(b_2)=4$.
 
\emph{Construction of $\Gopt$.} The basis of $U$ is $\{110\}$; we extend it to the basis $\{110,100,001\}$ of $\F_2^3$. So
\[
B_U=\begin{bmatrix}1&1&0\\1&0&0\\0&0&1\end{bmatrix},\qquad
B_U^{-1}=\begin{bmatrix}0&1&0\\1&1&0\\0&0&1\end{bmatrix},
\]
and the required generator matrix is
\[
\Gopt=B_U^{-1}B=
\begin{bmatrix}
0&0&1&1&1&0&1\\
0&1&1&1&1&0&0\\
1&1&1&0&0&1&0
\end{bmatrix}.
\]
Table~\ref{tab:ex} lists the encodings of all messages under $\Gopt$ and $G$, grouped by the cosets of $U$. Under $G$, the codeword $0100001$ of weight $2$ is the image of $010\notin U$, so $d_f(G)=2$, whereas under $\Gopt$ it is the image of $110\in U$ and $d_f(\Gopt)=4$. Both generator matrices generate the same code, with minimum distance $d=2$.
\end{example}
 
\begin{table}[t]
\centering
\caption{Encodings of the messages under $\Gopt$ and $G$, grouped by the cosets of $U=\{000,110\}$.}
\label{tab:ex}
\begin{tabular}{c|c|c}
\hline
Message $u$ & $u\Gopt$ & $uG$\\
\hline
$000$ & $0000000$ & $0000000$\\
$110$ & $0100001$ & $0111100$\\
\hline
$100$ & $0011101$ & $0011101$\\
$010$ & $0111100$ & $0100001$\\
\hline
$001$ & $1110010$ & $1110010$\\
$111$ & $1010011$ & $1001110$\\
\hline
$101$ & $1101111$ & $1101111$\\
$011$ & $1001110$ & $1010011$\\
\hline
$d$   & $2$ & $2$\\
$d_f$ & $4$ & $2$\\
\hline
\end{tabular}
\end{table}

\subsection{Generator Matrix in Systematic Form with Optimal $d_f$}\label{subsec:sys}

By Theorem~\ref{thm:opt}, the optimal value $d_f^{*}$ equals $\wt(b_{r+1})$ and is achieved by $\Gopt$. As shown in Example~\ref{ex:nosys}, $d_f^{*}$ need not be achieved by a generator matrix in systematic form. In this subsection, we identify when such a generator matrix exists.

We call
\[
W=\langle\, c\in C : \wt(c)<d_f^{*} \,\rangle
\]
the \textit{low-weight subcode} of $C$. It depends only on $C$ and $f$, and not on the generator matrix. The following lemma shows that a generator matrix is optimal exactly when its kernel subcode contains $W$.

\begin{lemma}\label{lem:W}
A generator matrix $G$ of $C$ achieves $d_f^{*}$ if and only if $W\subseteq D_G$.
\end{lemma}
\begin{proof}
If $W\subseteq D_G$, then every codeword of weight less than $d_f^{*}$ lies in $D_G$, so every $x\in C\setminus D_G$ has $\wt(x)\ge d_f^{*}$, and Lemma~\ref{lem:df} gives $d_f(G)\ge d_f^{*}$. Since $d_f(G)\le d_f^{*}$ by the definition of $d_f^{*}$, equality holds. Conversely, if $W\not\subseteq D_G$, then, since $D_G$ is a subspace and $W$ is spanned by codewords of weight less than $d_f^{*}$, one of these codewords $c$ is not in $D_G$, and Lemma~\ref{lem:df} gives $d_f(G)\le\wt(c)<d_f^{*}$.
\end{proof}

Let $G$ be a generator matrix of $C$. A set $\mathcal{I}\subseteq[n]$ of $k$ coordinates is called an information set of $C$ if the $k\times k$ submatrix of $G$ formed by the columns indexed by $\mathcal{I}$ is invertible; equivalently, the projection $c\mapsto (c_i)_{i\in\mathcal{I}}$ is a bijection from $C$ onto $\mathbb{F}_q^k$. Because systematic encodings assign message symbols to coordinates in a fixed order, we work with ordered information sets. An ordered information set of $C$ is a tuple $\mathcal{I}=(i_1,\dots,i_k)$ of distinct coordinates whose entries form an information set. We write $G_{\mathcal{I}}$ for the $k\times k$ matrix formed by the columns $i_1,\dots,i_k$ of $G$ in this order, and for $c\in C$ we write $c|_{\mathcal{I}}=(c_{i_1},\dots,c_{i_k})$. The matrix $G_{\mathcal{I}}^{-1}G$ is a generator matrix of $C$ whose columns $i_1,\dots,i_k$ form the identity matrix $I_k$.
Permuting these columns to the first $k$ positions gives a generator matrix in systematic form $[I_k\mid P]$ of a code equivalent to $C$. Conversely, every generator matrix in systematic form of a code obtained from $C$ by a permutation of the coordinates arises in this way. Since a permutation of the coordinates preserves the Hamming distance, it does not change $d_f$, and it is enough to consider the generator matrices $G_{\mathcal{I}}^{-1}G$. Under the encoding $u\mapsto uG_{\mathcal{I}}^{-1}G$, the message $u$ is mapped to the unique codeword $c\in C$ with $c|_{\mathcal{I}}=u$. For an ordered information set $\mathcal{I}$, we write $W|_{\mathcal{I}}=\{c|_{\mathcal{I}} : c\in W\}$ and denote by $D_{\mathcal{I}}$ the kernel subcode of $G_{\mathcal{I}}^{-1}G$.

\begin{theorem}\label{thm:sys}
There exists a generator matrix of $C$ in systematic form, up to a permutation of the coordinates, that achieves $d_f^{*}$ if and only if $W|_{\mathcal{I}}\subseteq U$ for some ordered information set $\mathcal{I}$ of $C$.
\end{theorem}
\begin{proof}
As shown above, it is enough to consider the generator matrices $G_{\mathcal{I}}^{-1}G$. We first show that
\[
D_{\mathcal{I}}=\{c\in C : c|_{\mathcal{I}}\in U\}.
\]
If $c=uG_{\mathcal{I}}^{-1}G$ for some $u\in U$, then $c|_{\mathcal{I}}=uG_{\mathcal{I}}^{-1}G_{\mathcal{I}}=u\in U$. Conversely, if $c\in C$ and $c|_{\mathcal{I}}\in U$, then $c$ and $(c|_{\mathcal{I}})G_{\mathcal{I}}^{-1}G$ are codewords with the same restriction to $\mathcal{I}$. Since $G_{\mathcal{I}}$ is invertible, the map $c\mapsto c|_{\mathcal{I}}$ is injective on $C$, so the two codewords are equal and $c\in D_{\mathcal{I}}$.

By Lemma~\ref{lem:W}, $G_{\mathcal{I}}^{-1}G$ achieves $d_f^{*}$ if and only if $W\subseteq D_{\mathcal{I}}$, which by the above holds if and only if $W|_{\mathcal{I}}\subseteq U$. In that case, $G_{\mathcal{I}}^{-1}G$, with the columns $i_1,\dots,i_k$ moved to the first $k$ positions, is a generator matrix in systematic form that achieves $d_f^{*}$.
\end{proof}

Since the map $c\mapsto c|_{\mathcal{I}}$ is linear and $U$ is a subspace, it is enough to check the condition of Theorem~\ref{thm:sys} on a basis of $W$. Such a basis is given by the greedy basis of Algorithm~\ref{alg:opt}, as shown next.

\begin{corollary}\label{cor:Wbasis}
Let $b_1,\dots,b_k$ be a basis of $C$ obtained by Algorithm~\ref{alg:opt}, and let $m$ be the number of indices $i$ with $\wt(b_i)<d_f^{*}$. Then $m\le r$ and $W=\langle b_1,\dots,b_m\rangle$. Hence, a generator matrix in systematic form achieving $d_f^{*}$ exists if and only if there is an ordered information set $\mathcal{I}$ of $C$ such that
\[
b_j|_{\mathcal{I}}\in U \quad\text{for all } 1\le j\le m .
\]
\end{corollary}
\begin{proof}
By \eqref{eq:mono} and $d_f^{*}=\wt(b_{r+1})$, the indices $i$ with $\wt(b_i)<d_f^{*}$ are $1,\dots,m$ with $m\le r$, and $\wt(b_{m+1})=d_f^{*}$. Clearly $\langle b_1,\dots,b_m\rangle\subseteq W$. Conversely, let $c\in C$ with $\wt(c)<d_f^{*}$. If $c\notin\langle b_1,\dots,b_m\rangle$, then at step $m+1$ of Algorithm~\ref{alg:opt} a codeword of weight at most $\wt(c)<\wt(b_{m+1})$ would have been chosen, which is a contradiction. Hence $W\subseteq\langle b_1,\dots,b_m\rangle$. The last statement follows from Theorem~\ref{thm:sys}.
\end{proof}

Note that the condition in Corollary~\ref{cor:Wbasis} involves only $b_1,\dots,b_m$ and not all of $b_1,\dots,b_r$.
When $m<r$, the codewords $b_{m+1},\dots,b_r$ have weight $d_f^{*}$, and they may lie outside $D_{\mathcal{I}}$ without reducing $d_f$, as shown in Example~\ref{ex:sys}.
 When $m=r$, since the map $c\mapsto c|_{\mathcal{I}}$ is injective on $C$ and $\dim U=r$, the condition becomes $U=\langle b_1|_{\mathcal{I}},\dots,b_r|_{\mathcal{I}}\rangle$.

The following necessary condition can be checked without searching over information sets.

\begin{corollary}\label{cor:nec}
If a generator matrix in systematic form achieves $d_f^{*}$, then
\[
\wt(c)\ge\min\{\wt(u) : u\in U\setminus\{0\}\}\quad\text{for all } c\in W\setminus\{0\}.
\]
\end{corollary}
\begin{proof}
Let $\mathcal{I}$ be an ordered information set with $W|_{\mathcal{I}}\subseteq U$, as in Theorem~\ref{thm:sys}, and let $c\in W\setminus\{0\}$. Since the map $c\mapsto c|_{\mathcal{I}}$ is injective on $C$, $c|_{\mathcal{I}}$ is a nonzero vector of $U$, and $\wt(c)\ge\wt(c|_{\mathcal{I}})$.
\end{proof}

\begin{example}\label{ex:sys}
Consider the $[7,3]$ code $C$ and the function $f$ with $U=\{000,110\}$ from Example~\ref{ex:73}. Here $d_f^{*}=4$, $m=1$ and $W=\langle b_1\rangle$ with $b_1=0100001$. For the ordered information set $\mathcal{I}=(2,7,1)$, we have $b_1|_{\mathcal{I}}=110\in U$. Therefore, by Corollary~\ref{cor:Wbasis}, the generator matrix
\[
G_{\mathcal{I}}^{-1}G=
\begin{bmatrix}
0&1&1&1&1&0&0\\
0&0&1&1&1&0&1\\
1&0&0&1&1&1&0
\end{bmatrix}
\]
achieves $d_f^{*}=4$. Moving the columns $2,7,1$ to the first three positions gives the generator matrix in systematic form
\[
\begin{bmatrix}
1&0&0&1&1&1&0\\
0&1&0&1&1&1&0\\
0&0&1&0&1&1&1
\end{bmatrix},
\]
for which $D_{\mathcal{I}}=\{0000000,0100001\}$ in the original coordinates. On the other hand, for $\mathcal{I}=(1,2,3)$ we have $b_1|_{\mathcal{I}}=010\notin U$, and the corresponding generator matrix in systematic form gives $d_f=2$.

Next, let $f(u_1,u_2,u_3)=u_1+u_2$, so that $U=\ker f=\{000,110,001,111\}$ and $r=2$. By Theorem~\ref{thm:opt}, $d_f^{*}=\wt(b_3)=4$, and since only $b_1$ has weight less than $4$, we have $m=1<r$ and again $W=\langle b_1\rangle$. For $\mathcal{I}=(2,7,1)$, we have $b_1|_{\mathcal{I}}=110\in U$, so the same generator matrix $G_{\mathcal{I}}^{-1}G$ achieves $d_f^{*}=4$, now with $D_{\mathcal{I}}=\{0000000,0100001,1001110,1101111\}$. 

Now consider instead a linear function with $U=\{000,111\}$. Then $W$ is again $\langle b_1\rangle$, but $\wt(b_1)=2<3=\min\{\wt(u) : u\in U\setminus\{0\}\}$. By Corollary~\ref{cor:nec}, no generator matrix in systematic form achieves $d_f^{*}=4$, and every such generator matrix gives $d_f=2$, whereas Algorithm~\ref{alg:opt} gives $d_f(\Gopt)=4$.

Similarly, in Example~\ref{ex:nosys}, we have $d_f^{*}=2$ and $W=\langle 100\rangle$, while $\min\{\wt(u) : u\in U\setminus\{0\}\}=2>\wt(100)$. Hence, Corollary~\ref{cor:nec} explains why no generator matrix in systematic form achieves $d_f^{*}$ in that example.
\end{example}

\arxiv{
\section{Non-systematic Encodings for Functions with Two Values}\label{sec:twoval}

Section~\ref{sec:main} considered linear encodings of a fixed linear code. In this section, the encoding is not required to be linear or systematic. A one-to-one map $\mathrm{Enc}:\F_2^k\to\F_2^n$ that satisfies~\eqref{eq:fcc} is called a non-systematic $(f,t)$-FCC of length $n$, and $n-k$ is its redundancy. Examples~\ref{ex:or} and~\ref{ex:lin2} are of this kind. We consider a nonconstant $f:\F_2^k\to\F_2$ and write $M=\min\{|f^{-1}(0)|,|f^{-1}(1)|\}$, so that the two classes have sizes $M$ and $2^k-M$.

For $S\subseteq\F_2^n$ and $\rho\ge 0$, let $N_\rho(S)=\{y\in\F_2^n : d(y,S)\le\rho\}$ denote the radius-$\rho$ neighbourhood of $S$,  and let $B_j$ be the Hamming ball of radius $j$ around $0$, so that $|B_j|=V(n,j)$ and $B_j=N_j(\{0\})$. If $\mathrm{Enc}$ maps the two classes onto $S,T\subseteq\F_2^n$, then~\eqref{eq:fcc} holds if and only if $d(S,T)\ge 2t+1$, that is, $T\cap N_{2t}(S)=\emptyset$. Since $|T|=2^k-|S|$, the question is how small $N_{2t}(S)$ can be for a set $S$ of given size. This is the vertex-isoperimetric problem on the hypercube: among all sets $S\subseteq\F_2^n$ of a given size, find one whose neighbourhood $N_\rho(S)$ is as small as possible. Harper~\cite{H1966} showed that Hamming balls are optimal, and, for sizes between two ball sizes, a ball together with a suitably chosen part of the next weight vectors.

Following~\cite{R2020}, we define the simplicial order on $\F_2^n$ by $x\prec y$ if $\wt(x)<\wt(y)$, or $\wt(x)=\wt(y)$ and $\min\big(\mathrm{supp}(x+y)\big)\in\mathrm{supp}(x)$. Thus, vectors are ordered by weight, and vectors of equal weight are compared at the first coordinate where they differ. For $1\le M\le 2^n$, let $I_M\subseteq\F_2^n$ be the set of the first $M$ vectors in this order. We call $I_M$ the initial segment of size $M$ of the simplicial order. For example, in $\F_2^4$,
\[
I_7=\{0000,1000,0100,0010,0001,1100,1010\}.
\]
If $M=V(n,j)$, then $I_M=B_j$.

\begin{lemma}[{\cite{H1966}, \cite[Th.~1, Th.~4]{R2020}}]\label{lem:harper}
For every $S\subseteq\F_2^n$ with $|S|=M$ and every $\rho\ge 1$,
\[
|N_\rho(S)|\ge|N_\rho(I_M)|.
\]
\end{lemma}
For $\rho=1$, Lemma~\ref{lem:harper} is Harper's theorem~\cite{H1966},~\cite[Th.~1]{R2020}: no subset of $\F_2^n$ of size $M$ has a smaller neighbourhood than $I_M$. R\"aty~\cite[Th.~4]{R2020} showed that this minimality carries over to larger radius, in the sense that a set whose neighbourhood has minimal size among all sets of the same size also minimizes the size of its $\rho$-neighbourhood for every $\rho\ge 1$. Taking this set to be $I_M$ gives the lemma for all $\rho$.

\begin{theorem}\label{thm:twoval}
Let $f:\F_2^k\to\F_2$ be a non-constant function, let $t\ge 1$, and let $n\ge k$. A non-systematic $(f,t)$-FCC of length $n$ exists if and only if
\begin{equation}\label{eq:twoval}
|N_{2t}(I_M)|\le 2^n-2^k+M,
\end{equation}
where $I_M\subseteq\F_2^n$.
\end{theorem}
\begin{proof}
Suppose that $\mathrm{Enc}$ is a non-systematic $(f,t)$-FCC of length $n$, and let $S$ and $T$ be the images of the classes of sizes $M$ and $2^k-M$. Since $T\cap N_{2t}(S)=\emptyset$, Lemma~\ref{lem:harper} gives
\[
2^k-M=|T|\le 2^n-|N_{2t}(S)|\le 2^n-|N_{2t}(I_M)|,
\]
which is~\eqref{eq:twoval}. Conversely, if~\eqref{eq:twoval} holds, map the class of size $M$ bijectively onto $I_M$ and the other class injectively into $\F_2^n\setminus N_{2t}(I_M)$, which has at least $2^k-M$ elements. The two images are disjoint and at distance at least $2t+1$, so the map is one-to-one and satisfies~\eqref{eq:fcc}.
\end{proof}

The left side of~\eqref{eq:twoval} can be computed from $M$ alone. Write $M=V(n,j)+a$ with $0\le a<\binom{n}{j+1}$. Then $I_M=B_j\cup L_a$, where $L_a$ consists of the first $a$ vectors of weight $j+1$ in the simplicial order. Every vector of weight at most $j+\rho$ is within distance $\rho$ of $B_j$. A vector $y$ of weight $j+\rho+1$ is at distance at least $\rho+1$ from $B_j$, and it is within distance $\rho$ of a vector $x\in L_a$ exactly when $y$ is obtained from $x$ by changing $\rho$ zeros into ones, that is, when $\mathrm{supp}(x)\subseteq\mathrm{supp}(y)$. Vectors of larger weight are at a distance more than $\rho$ from $I_M$. Hence
\[
N_\rho(I_M)=B_{j+\rho}\cup\big\{y\in\F_2^n : \wt(y)=j+\rho+1,\ \mathrm{supp}(x)\subseteq\mathrm{supp}(y)\text{ for some } x\in L_a\big\}.
\]

If $M=V(n,j)$, then $L_a=\emptyset$ and $N_\rho(I_M)=B_{j+\rho}$. This gives the following sufficient condition.

\begin{corollary}\label{cor:ball}
If $M\le V(n,j)$ and $2^k-M\le 2^n-V(n,j+2t)$ for some $j\ge 0$, then a non-systematic $(f,t)$-FCC of length $n$ exists.
\end{corollary}
\begin{proof}
By assumption, $|B_j|\ge M$ and $|\F_2^n\setminus B_{j+2t}|=2^n-V(n,j+2t)\ge 2^k-M$. Let $\mathrm{Enc}$ map the class of size $M$ injectively into $B_j$ and the other class injectively into $\F_2^n\setminus B_{j+2t}$. Since these two sets are disjoint, $\mathrm{Enc}$ is injective. If $f(u)\neq f(v)$, one of $\mathrm{Enc}(u)$ and $\mathrm{Enc}(v)$ has weight at most $j$ and the other has weight at least $j+2t+1$, so $d\big(\mathrm{Enc}(u),\mathrm{Enc}(v)\big)\ge 2t+1$.
\end{proof}

Examples \ref{ex:or} and \ref{ex:lin2} of Section~\ref{intro} are special cases of Corollary~\ref{cor:ball}. Example~\ref{ex:or} is the case $M=1$, $j=0$, where the condition becomes $2^k-1\le 2^n-V(n,2t)$, which for $k=4$, $t=1$, $n=5$ reads $15\le 16$. Example~\ref{ex:lin2} is the case $k=8$, $n=9$, $t=1$, $j=3$, with $M=128\le V(9,3)=130$ and $128\le 2^9-V(9,5)=130$.

We now take $t=1$ and determine when a single redundant symbol is sufficient. By the argument following Example~\ref{ex:or}, every non-systematic $(f,1)$-FCC has redundancy at least $1$. Theorem~\ref{thm:twoval} with $n=k+1$ gives the following.

\begin{corollary}\label{cor:r1}
Let $f:\F_2^k\to\F_2$ be nonconstant. A non-systematic $(f,1)$-FCC with redundancy $1$ exists if and only if $|N_2(I_M)|\le 2^k+M$, where $I_M\subseteq\F_2^{k+1}$.
\end{corollary}

The condition depends on $f$ only through $M$, not on how the two classes are placed in $\F_2^k$. Table~\ref{tab:r1} lists the values of $M$ that satisfy it for $k\le 9$. The values were obtained by computing $|N_2(I_M)|$ for every $M\le 2^{k-1}$ as described in Appendix~\ref{app:k10}, and the row $k=7$ is worked out by hand in Appendix~\ref{app:k7}. The set of such $M$ is not an interval for $k=7,8,9$. For example, on $\F_2^7$ the class sizes $(36,92)$ allow redundancy $1$, by Corollary~\ref{cor:ball} with $j=2$, while $(30,98)$ do not. A nonzero linear $f$ has two classes of size $2^{k-1}$, so $M=2^{k-1}$. By Table~\ref{tab:r1}, redundancy $1$ is then possible if and only if $k\ge 8$, so Example~\ref{ex:lin2} is the smallest such case.

\begin{table}[t]
\centering
\caption{Values of $M\le 2^{k-1}$ for which a non-systematic $(f,1)$-FCC with redundancy $1$ exists}
\label{tab:r1}
\begin{tabular}{cl}
\hline
$k$ & $M$ \\
\hline
$2,3$ & none \\
$4$ & $1$ \\
$5$ & $1,2$ \\
$6$ & $1,\dots,8$ \\
$7$ & $1,\dots,22$ and $35,36,37$ \\
$8$ & $1,\dots,59$ and $70,\dots,74$ and $126,127,128$ \\
$9$ & $1,\dots,204$ and $252,\dots,256$ \\
$\ge 10$ & all $M$ \\
\hline
\end{tabular}
\end{table}

\begin{proposition}\label{prop:k10}
For $k\ge 10$, every nonconstant $f:\F_2^k\to\F_2$ has a non-systematic $(f,1)$-FCC with redundancy $1$, which is optimal.
\end{proposition}
\begin{proof}
By Corollary~\ref{cor:r1}, it is sufficient to show that
\begin{equation}\label{eq:r1cond}
|N_2(I_M)|\le 2^k+M
\end{equation}
for every $M\le 2^{k-1}$, where $I_M\subseteq\F_2^{k+1}$. For $10\le k\le 21$, we verified~\eqref{eq:r1cond} directly by computing $|N_2(I_M)|$ for every $M\le 2^{k-1}$; see Appendix~\ref{app:k10}.

Now let $k\ge 22$ and $n=k+1$, and let $j$ be the least integer with $M\le V(n,j)$, so that $M\ge V(n,j-1)+1$, where $V(n,-1)=0$. Since $I_M\subseteq B_j$, we have $N_2(I_M)\subseteq B_{j+2}$, and~\eqref{eq:r1cond} holds whenever $V(n,j+2)\le 2^k+M$. Since $M\ge V(n,j-1)+1$, this is implied by $V(n,j+2)-V(n,j-1)\le 2^k+1$, that is, by
\begin{equation}\label{eq:binom3}
\binom{k+1}{j}+\binom{k+1}{j+1}+\binom{k+1}{j+2}\le 2^k+1.
\end{equation}
The left side of~\eqref{eq:binom3} is at most $3\binom{k+1}{\lfloor (k+1)/2\rfloor}$, which is at most $2^k$ for $k\ge 22$, as shown in Appendix~\ref{app:k10}. Hence~\eqref{eq:binom3} holds for every $j$, and~\eqref{eq:r1cond} follows.
\end{proof}

\subsection{Redundancy Thresholds for General $t$}

Theorem~\ref{thm:twoval} and Corollary~\ref{cor:ball} hold for every $t\ge 1$, but Corollary~\ref{cor:r1}, Table~\ref{tab:r1} and Proposition~\ref{prop:k10} are for $t=1$ only. For $t=1$, Proposition~\ref{prop:k10} combines an analytic bound for $k\ge 22$ with a direct computation over all $M\le 2^{k-1}$ for $10\le k\le 21$. For larger $t$, the analytic bound applies only for larger $k$, and the direct computation over all $M$ is no longer feasible below that bound. In this subsection, we show that, for every $t$, redundancy $1$ is achievable for all nonconstant two-valued functions once $k$ is large enough. For smaller $k$ we introduce, for each redundancy $s\le 2t$, the least $\kappa$ such that, for every $k\ge\kappa$, redundancy $s$ suffices for all such functions, and we give bounds on it. We use the following bound on the central binomial coefficient.

\begin{lemma}\label{lem:central}
Let $c_n$ denote the highest term in the binomial expansion of $2^n$, i.e., $c_n=\binom{n}{\lfloor n/2\rfloor}$. For every $n\ge 1$,
$c_n\le 2^n\sqrt{2/(3n)}$.
\end{lemma}

\begin{proof}
The proof is given in the Appendix \ref{app:C}.
\end{proof}

\begin{theorem}\label{thm:larget}
Let $t\ge 1$ and
\[
k\ \ge\ \Big\lceil \tfrac{8}{3}(2t+1)^2\Big\rceil-1 .
\]
Then every nonconstant $f:\mathbb{F}_2^k\to \mathbb{F}_2$ has a non-systematic $(f,t)$-FCC of length $n=k+1$. The redundancy $n-k=1$ is
optimal.
\end{theorem}

\begin{proof}
Let $M\le 2^{k-1}$ be the size of the smaller class and put $n=k+1$. By Theorem~\ref{thm:twoval} it suffices to show
\begin{equation}\label{eq:largetcond}
|N_{2t}(I_M)|\le 2^n-2^k+M=2^k+M .
\end{equation}
Let $j\ge 0$ be the least integer with $V(n,j)\ge M$, and set $V(n,-1)=0$. Then $M\ge V(n,j-1)+1$. Since the simplicial order lists
words by weight, $I_M\subseteq B_j$, and hence
$N_{2t}(I_M)\subseteq N_{2t}(B_j)=B_{j+2t}$. Therefore,
\[
|N_{2t}(I_M)|\le V(n,j+2t)
= V(n,j-1)+\sum_{i=j}^{j+2t}\binom{n}{i}
\le M-1+(2t+1)\,c_n ,
\]
with $\binom{n}{i}=0$ for $i>n$. So \eqref{eq:largetcond} holds as soon
as $(2t+1)c_n -1 \le 2^k=2^{n-1}$, which is equivalent to $c_n \leq \frac{2^{n-1}+1}{2t+1}$. 
By Lemma~\ref{lem:central}, $c_n\le 2^n\sqrt{2/(3n)}$, so it suffices
that
\[
2^n\sqrt{\frac{2}{3n}}\ \le\ \frac{2^{n-1}}{2t+1}.
\]
This is equivalent to $2(2t+1)\sqrt{2/(3n)}\le 1$, that is,
$n\ge\frac83(2t+1)^2$. As $n=k+1$, the previous condition is equivalent to $k\ge\lceil\frac83(2t+1)^2\rceil-1$. Optimality of redundancy $1$ follows from $n-k\ge 1$, established after Example~\ref{ex:or}.
\end{proof}

Theorem~\ref{thm:larget} shows that for every $t\ge 1$ redundancy $1$ is
achievable for all nonconstant two-valued functions when $k$ is large
enough. The bound on $k$ in the theorem is not sharp. For $t=1$ the exact
value follows from Proposition~\ref{prop:k10}, and for $t\ge 2$ it is
unknown. For smaller $k$, some functions require redundancy larger than
$1$. However, redundancy $2t$ is always sufficient: for every two-valued
function, $\mathrm{Enc}(u)=(u,f(u)^{2t})$, where $f(u)^{2t}$ denotes the
value $f(u)$ repeated $2t$ times, is a systematic $(f,t)$-FCC with
redundancy $2t$, since two messages with different function values are
at distance at least $1$ in the systematic part and at distance $2t$ in
the parity part, and this redundancy is optimal among systematic
encodings by~\cite[Cor.~1]{LBWY2023}. A systematic encoding
is in particular a non-systematic one, so for every $k$ and every
nonconstant two-valued $f$ there is a non-systematic $(f,t)$-FCC of length
$k+2t$.

\begin{definition}\label{def:threshold}
Let $t\ge 1$ and $1\le s\le 2t$. The \emph{redundancy-$s$ threshold}
$\kappa_s(t)$ is the least positive integer $\kappa$ such that, for every
$k\ge\kappa$, every nonconstant function $f:\mathbb{F}_2^k\to \mathbb{F}_2$
admits a non-systematic $(f,t)$-FCC of length $k+s$.
\end{definition}

Since every nonconstant two-valued function on $\mathbb{F}_2^k$ admits a
non-systematic $(f,t)$-FCC of length $k+2t$ for every $k\ge 1$, as shown
above, we have $\kappa_{2t}(t)=1$. Appending a zero coordinate
to every codeword of an $(f,t)$-FCC of length $n$ gives an $(f,t)$-FCC of
length $n+1$, so $\kappa_s(t)$ exists for every $1\le s\le 2t$ and
\[
\kappa_1(t)\ \ge\ \kappa_2(t)\ \ge\ \cdots\ \ge\ \kappa_{2t}(t)=1 .
\]
In particular, if $\kappa_{s+1}(t)\le k<\kappa_s(t)$ for some
$1\le s\le 2t-1$, every two-valued function on $\mathbb{F}_2^k$ has a
non-systematic encoding with redundancy at most $s+1$.

A lower bound on $\kappa_1(t)$ comes from a function whose smaller class
has a single element, such as the OR function. For $M=1$, we have $I_1=\{0\}$ and
$N_{2t}(I_1)=B_{2t}$, so with $n=k+1$, Theorem \ref{thm:twoval} gives $V(k+1,2t)\le 2^k+1$ or equivalently $V(n,2t)\leq 2^{n-1}+1$. Suppose $k\le 4t-2$.
Then $n\le 4t-1$, so $2t\ge\lfloor n/2\rfloor+1$. Since
$V(n,\lfloor n/2\rfloor)\ge 2^{n-1}$, with equality only for odd $n$, and
$\binom{n}{\lfloor n/2\rfloor+1}\ge 1$, with strict inequality for odd
$n\ge 3$,
\begin{align*}
V(n,2t)\ &\geq  V(n,\lfloor n/2\rfloor+1) \qquad \big(\text{as }2t\ge\lfloor n/2\rfloor+1 \big) \\
&=\ V(n,\lfloor n/2\rfloor)+\binom{n}{\lfloor n/2\rfloor+1}\\
& >\ 2^{n-1}+1.
\end{align*}
Redundancy $1$ is therefore impossible for this function when
$k\le 4t-2$, and, together with Theorem~\ref{thm:larget},
\[
4t-1\ \le\ \kappa_1(t)\ \le\ \Big\lceil \tfrac{8}{3}(2t+1)^2\Big\rceil-1 .
\]
For $t=1$ these bounds give $3\le\kappa_1(1)\le 23$, while
Proposition~\ref{prop:k10} and Table~\ref{tab:r1} give
$\kappa_1(1)=10$. The lower bound is linear in $t$ and the upper bound
quadratic, and the exact growth of $\kappa_1(t)$ remains open.

The next theorem extends the bound of Theorem~\ref{thm:larget} from redundancy $1$ to every redundancy $s<2t$.

\begin{theorem}\label{thm:kappa}
Let $t\ge 1$ and $1\le s\le 2t-1$. Then
\[
\kappa_s(t)\ \le\ \left\lceil\frac{2(2t+1)^2}{3\,(1-2^{-s})^2}\right\rceil-s .
\]
\end{theorem}

\begin{proof}
Let $k$ be at least the right-hand side, let
$f:\mathbb{F}_2^k\to\mathbb{F}_2$ be nonconstant with the smaller class of size $M$, and let $n=k+s$. The argument in the proof of Theorem~\ref{thm:larget} gives $|N_{2t}(I_M)|\le M-1+(2t+1)c_n$. By
Theorem~\ref{thm:twoval}, a non-systematic $(f,t)$-FCC of length $n$
exists if $(2t+1)c_n-1\le 2^n-2^k=2^n(1-2^{-s})$, which is equivalent to 
$$c_n \leq \frac{2^n(1-2^{-s})+1}{2t+1}.$$
By Lemma~\ref{lem:central}, this is implied by $2^n\sqrt{2/(3n)} \leq \frac{2^n(1-2^{-s})}{2t+1}$, 
that is,
$n\ge\frac{2(2t+1)^2}{3(1-2^{-s})^2}$. As $n=k+s$, the previous condition is equivalent to $k\ge\frac{2(2t+1)^2}{3(1-2^{-s})^2}-s$.
Since $k$ was an arbitrary integer with
$k\ge\big\lceil\frac{2(2t+1)^2}{3(1-2^{-s})^2}\big\rceil-s$,
Definition~\ref{def:threshold} gives the claim.
\end{proof}

For $s=1$, Theorem~\ref{thm:kappa} gives
$\kappa_1(t)\le\lceil\frac83(2t+1)^2\rceil-1$, which is the bound of
Theorem~\ref{thm:larget}. For $t=2$ it gives $\kappa_1(2)\le 66$ and
$\kappa_2(2)\le 28$. Lower bounds on $\kappa_s(t)$ for $2\le s\le 2t-1$
remain open.
}

\section{Conclusion}
\label{conclusion}

For a fixed linear code $C$ and a linear function $f$, the generator matrices of $C$ are not equivalent as function-correcting codes. We showed that the largest attainable separation $d_f^{*}$ equals $\wt(b_{r+1})$, where $b_1,\dots,b_k$ is the greedy basis of $C$ and $r=\dim\ker f$, and that it is attained by any generator matrix whose kernel subcode is $\langle b_1,\dots,b_r\rangle$. Whether $d_f^{*}$ can be attained in systematic form depends on the existence of an information set $\mathcal{I}$ with $W|_{\mathcal{I}}\subseteq\ker f$, where $W$ is spanned by the codewords of weight less than $d_f^{*}$, and a necessary condition compares the weights of these codewords with the minimum weight of $\ker f$.

\arxiv{For two-valued functions, the existence of a non-systematic $(f,t)$-FCC of a given length reduces to the size of the radius-$2t$ neighbourhood of an initial segment of the simplicial order, and so depends on $f$ only through the size of its smaller preimage. For $t=1$, this determines the optimal redundancy for every $k$, with redundancy $1$ sufficient for all $k\ge 10$. For general $t$ we obtained the upper bound $\kappa_s(t)\le\lceil 2(2t+1)^2/(3(1-2^{-s})^2)\rceil-s$ on the redundancy-$s$ threshold and the lower bound $\kappa_1(t)\ge 4t-1$. The gap between linear and quadratic growth in $t$ is open, as are lower bounds on $\kappa_s(t)$ for $s\ge 2$.}
\section{Acknowledgement}
C. Rajput thanks the University of Bristol Mumbai Enterprise Campus for its support.

\arxiv{
\appendices

\section{The Case $k=7$ in Table~\ref{tab:r1}}\label{app:k7}

In this appendix we verify the row $k=7$ of Table~\ref{tab:r1}. Here $n=8$, and by Corollary~\ref{cor:r1} a class size $M\le 64$ allows redundancy $1$ if and only if
\begin{equation}\label{eq:k7}
|N_2(I_M)|\le 128+M,
\end{equation}
where $I_M\subseteq\F_2^8$. We use $V(8,1)=9$, $V(8,2)=37$, $V(8,3)=93$ and $V(8,4)=163$.

We say that a vector $y$ covers a vector $x$ if $\mathrm{supp}(x)\subseteq\mathrm{supp}(y)$. Write $M=V(8,j)+a$ with $0\le a<\binom{8}{j+1}$, so that $I_M=B_j\cup L_a$, where $L_a$ consists of the first $a$ vectors of weight $j+1$ in the simplicial order. As shown after Theorem~\ref{thm:twoval}, $N_2(I_M)$ is the union of $B_{j+2}$ and the set of vectors of weight $j+3$ that cover some vector of $L_a$. Let $s_j(a)$ denote the number of such vectors of weight $j+3$. Then
\[
|N_2(I_M)|=V(8,j+2)+s_j(a).
\]
We identify vectors with their supports and list the vectors of weight $j+1$ in lex order. Table~\ref{tab:k7} gives $|N_2(I_M)|$ for selected values of $M$.

\begin{table}[t]
\centering
\caption{The condition $|N_2(I_M)|\le 128+M$ for $k=7$ and selected $M$}
\label{tab:k7}
\begin{tabular}{cccc}
\hline
$M$ & $|N_2(I_M)|$ & $128+M$ & Redundancy $1$ \\
\hline
$1$ & $37$ & $129$ & yes \\
$9$ & $93$ & $137$ & yes \\
$16$ & $128$ & $144$ & yes \\
$22$ & $148$ & $150$ & yes \\
$23$ & $154$ & $151$ & no \\
$30$ & $162$ & $158$ & no \\
$34$ & $163$ & $162$ & no \\
$35$ & $163$ & $163$ & yes \\
$36$ & $163$ & $164$ & yes \\
$37$ & $163$ & $165$ & yes \\
$38$ & $173$ & $166$ & no \\
$64$ & $211$ & $192$ & no \\
\hline
\end{tabular}
\end{table}

\emph{Case $1\le M\le 9$.} Here $I_M\subseteq B_1$, so $N_2(I_M)\subseteq B_3$ and $|N_2(I_M)|\le 93<128+M$. Condition~\eqref{eq:k7} holds.

\emph{Case $10\le M\le 22$.} Here $j=1$ and $1\le a\le 13$, so $L_a$ consists of pairs. In lex order, the first seven pairs are $\{1,2\},\dots,\{1,8\}$, and the next six are $\{2,3\},\dots,\{2,8\}$. A $4$-set covers one of the first $a\le 7$ pairs if and only if it contains $1$ and meets $\{2,\dots,a+1\}$. For $a=7+b$ with $1\le b\le 6$, it covers one of the first $a$ pairs if and only if it contains $1$, or it avoids $1$, contains $2$ and meets $\{3,\dots,b+2\}$. Hence
\[
s_1(a)=
\begin{cases}
35-\binom{7-a}{3}, & 1\le a\le 7,\\[2pt]
55-\binom{6-b}{3}, & a=7+b,\ 1\le b\le 6.
\end{cases}
\]
Condition~\eqref{eq:k7} reads $93+s_1(a)\le 137+a$, that is, $s_1(a)\le 44+a$. For $a\le 7$ this holds since $s_1(a)\le 35$. For $a=7+b$ it becomes $\binom{6-b}{3}\ge 4-b$, which holds for every $1\le b\le 6$, with equality at $b=3$ and $b=4$, that is, $M=19$ and $M=20$. All the covered $4$-sets contain $1$ or $2$, and there are only $55$ such sets, so the neighbourhood grows slowly in this range.

\emph{Case $23\le M\le 34$.} Here $j=1$ and $14\le a\le 25$. The pairs that follow are $\{3,4\},\dots,\{3,8\}$, then $\{4,5\},\dots,\{4,8\}$, then $\{5,6\},\{5,7\},\{5,8\}$. Each new pair covers $4$-sets avoiding the coordinates already used, which were not covered before. For $a=13+c$ with $1\le c\le 5$, for $a=18+d$ with $1\le d\le 4$, and for $a=22+e$ with $1\le e\le 3$, the same counting gives
\[
s_1(a)=65-\binom{5-c}{3},\qquad s_1(a)=69-\binom{4-d}{3},\qquad s_1(a)=70-\binom{3-e}{3},
\]
respectively. In each range $s_1(a)>44+a$. For example, at $M=23$ the pair $\{3,4\}$ adds the $\binom{4}{2}=6$ sets $\{3,4,x,y\}$ with $x,y\in\{5,\dots,8\}$, while the right side of~\eqref{eq:k7} increases by $1$. At $M=34$, all $70$ sets of size $4$ are covered and $|N_2(I_{34})|=163>162$. Condition~\eqref{eq:k7} fails throughout.

\emph{Case $M\in\{35,36,37\}$.} For $M=35,36$, all $4$-sets are still covered, so $N_2(I_M)=B_4$ and $|N_2(I_M)|=163$. The right side of~\eqref{eq:k7} is $163$ and $164$, so the condition holds. For $M=37$, $I_{37}=B_2$ and $N_2(I_{37})=B_4$, so $163\le 165$. This case is Corollary~\ref{cor:ball} with $j=2$: the class of size $37$ is mapped into the vectors of weight at most $2$, and the class of size $91$ into the $93$ vectors of weight at least $5$.

\emph{Case $38\le M\le 64$.} Here $j=2$ and $1\le a\le 27$, and condition~\eqref{eq:k7} reads $163+s_2(a)\le 165+a$, that is, $s_2(a)\le a+2$. Since $L_a$ grows with $a$, so does $s_2(a)$. The first $3$-set $\{1,2,3\}$ is covered by $\binom{5}{2}=10$ sets of size $5$, so $s_2(a)\ge 10>a+2$ for $a\le 6$. The first six $3$-sets are $\{1,2,x\}$ with $3\le x\le 8$, and every $5$-set containing $\{1,2\}$ covers one of them, so $s_2(a)\ge\binom{6}{3}=20>a+2$ for $7\le a\le 17$. The first $18$ $3$-sets are those containing $1$ and some element of $\{2,3,4,5\}$. Every $5$-set containing $1$ has four further elements, at most three of which lie in $\{6,7,8\}$, so it covers one of these $3$-sets. Hence $s_2(a)\ge\binom{7}{4}=35>a+2$ for $18\le a\le 27$. Condition~\eqref{eq:k7} fails throughout. In particular, it fails for $M=64$, so no nonzero linear function on $\F_2^7$ admits a non-systematic $(f,1)$-FCC with redundancy $1$.

Together, the cases show that redundancy $1$ is possible for $k=7$ exactly when $M\in\{1,\dots,22\}\cup\{35,36,37\}$. The gap $23\le M\le 34$ arises from the order in which the second layer fills. While the added pairs share coordinate $1$ or $2$, the $4$-sets they cover overlap heavily. Once the order reaches pairs on new coordinates, each pair covers several new $4$-sets, and the bound catches up only after the whole layer of weight $4$ is covered.

\section{Details for the Proof of Proposition~\ref{prop:k10}}\label{app:k10}

Throughout this appendix, $n=k+1$ and
\[
S_n(j)=\binom{n}{j}+\binom{n}{j+1}+\binom{n}{j+2},
\]
so that~\eqref{eq:binom3} reads $S_n(j)\le 2^k+1$.

\subsection{The case $k\ge 22$}

For every $0\le i\le n$ we have $\binom{n}{i}\le c_n$, where $c_n=\binom{n}{\lfloor n/2\rfloor}$ is the central binomial coefficient. Applying this to each of the three terms of $S_n(j)$ gives $S_n(j)\le 3c_n$ for every $j$. By Pascal's rule,
\[
c_{n+1}=\binom{n}{\lfloor (n+1)/2\rfloor}+\binom{n}{\lfloor (n+1)/2\rfloor-1}\le 2c_n,
\]
and dividing by $2^{n+1}$ gives $c_{n+1}/2^{n+1}\le c_n/2^n$. Hence $c_n/2^n\le c_{23}/2^{23}$ for every $n\ge 23$. For $n=23$,
\[
3c_{23}=3\binom{23}{11}=4056234\le 4194304=2^{22},
\]
that is, $3c_{23}/2^{23}\le 1/2$. Therefore $3c_n\le 2^{n-1}=2^k$ for every $n\ge 23$, and since $S_n(j)\le 3c_n$, inequality~\eqref{eq:binom3} holds for every $j$ when $k\ge 22$.

\subsection{The case $10\le k\le 21$}

For each $10\le k\le 21$, we computed $|N_2(I_M)|$ for every $M\le 2^{k-1}$ by listing the vectors of $\F_2^{n}$ in the simplicial order and adding the radius-$2$ ball of each new vector to the neighbourhood. In every case $|N_2(I_M)|\le 2^k+M$, so condition~\eqref{eq:r1cond} holds. Table~\ref{tab:k10} lists, for each $k$, the value of $M$ for which $2^k+M-|N_2(I_M)|$ is smallest. The same computation for $k\le 9$ gives the rows $k\le 9$ of Table~\ref{tab:r1}.

\begin{table}[t]
\centering
\caption{The closest case of $|N_2(I_M)|\le 2^k+M$ for $10\le k\le 21$}
\label{tab:k10}
\begin{tabular}{cccc}
\hline
$k$ & $M$ & $|N_2(I_M)|$ & $2^k+M$ \\
\hline
$10$ & $388$ & $1375$ & $1412$ \\
$11$ & $950$ & $2861$ & $2998$ \\
$12$ & $1744$ & $5371$ & $5840$ \\
$13$ & $4096$ & $11129$ & $12288$ \\
$14$ & $7597$ & $21091$ & $23981$ \\
$15$ & $16259$ & $42194$ & $49027$ \\
$16$ & $32439$ & $83113$ & $97975$ \\
$17$ & $65371$ & $161376$ & $196443$ \\
$18$ & $131070$ & $320931$ & $393214$ \\
$19$ & $230853$ & $590485$ & $755141$ \\
$20$ & $494803$ & $1207151$ & $1543379$ \\
$21$ & $972315$ & $2348105$ & $3069467$ \\
\hline
\end{tabular}
\end{table}

\section{Proof of Lemma \ref{lem:central}} \label{app:C}

We first show that
\begin{equation}\label{eq:central-even}
\binom{2m}{m}\le\frac{4^m}{\sqrt{3m+1}}\qquad\text{for all } m\ge 1
\end{equation}
by induction on $m$. For $m=1$, both sides equal $2$. Assume
\eqref{eq:central-even} holds for $m=\ell$, where $\ell\ge 1$. Then
\[
\binom{2\ell+2}{\ell+1}=\frac{(2\ell+2)!}{(\ell+1)!\,(\ell+1)!}
=\frac{(2\ell+2)(2\ell+1)}{(\ell+1)^2}\binom{2\ell}{\ell}
=\frac{2(2\ell+1)}{\ell+1}\binom{2\ell}{\ell}
\le\frac{2(2\ell+1)}{\ell+1}\cdot\frac{4^\ell}{\sqrt{3\ell+1}} .
\]
To prove \eqref{eq:central-even} for $m=\ell+1$, it remains to show
that the right-hand side is at most $4^{\ell+1}/\sqrt{3\ell+4}$, that is,
\[
\frac{2\ell+1}{\ell+1}\ \le\ 2\sqrt{\frac{3\ell+1}{3\ell+4}} .
\]
Squaring both sides, this is equivalent to
$(2\ell+1)^2(3\ell+4)\le 4(\ell+1)^2(3\ell+1)$. Expanding, the left-hand
side is $12\ell^3+28\ell^2+19\ell+4$ and the right-hand side is
$12\ell^3+28\ell^2+20\ell+4$, so the difference is $\ell\ge 0$. This
proves \eqref{eq:central-even}.

Let $n=2m$ be even. Then by \eqref{eq:central-even},
\[
c_n=\binom{2m}{m}\le\frac{4^m}{\sqrt{3m+1}}<\frac{2^n}{\sqrt{3n/2}}
=2^n\sqrt{\frac{2}{3n}} .
\]
Let $n=2m+1$ be odd. Since
$\binom{2m+2}{m+1}=\binom{2m+1}{m}+\binom{2m+1}{m+1}=2\binom{2m+1}{m}$,
applying \eqref{eq:central-even} with $m+1$ in place of $m$ gives
\[
c_n=\binom{2m+1}{m}=\frac12\binom{2m+2}{m+1}
\le\frac{2^{2m+1}}{\sqrt{3m+4}}
=\frac{2^n}{\sqrt{(3n+5)/2}}
<2^n\sqrt{\frac{2}{3n}} .
\]
}

\end{document}